\documentclass[review]{elsarticle}

\usepackage{lineno,hyperref}
\modulolinenumbers[5]

\journal{Journal of \LaTeX\ Templates}

\usepackage{fullpage}
\usepackage[dvipsnames,usenames]{color}

\usepackage{url}
\usepackage{epsfig}
\usepackage{epstopdf}
\usepackage{graphicx}
\DeclareGraphicsExtensions{.pdf,.png,.jpg,.eps}
\usepackage{amsthm}
\usepackage{latexsym}
\usepackage{amsfonts}
\usepackage{amssymb}
\usepackage{amsmath}
\usepackage{bm}
\usepackage{pseudocode}
\usepackage[ruled,linesnumbered]{algorithm2e}
\usepackage{array}
\usepackage{color}
\usepackage{multirow}
\usepackage{caption} 
\usepackage{tikz}
\usepackage{algpseudocode}

\newtheorem*{theorem*}{Theorem}
\newtheorem*{lemma*}{Lemma}
\newtheorem{lemma}{Lemma}
\newtheorem{theorem}{Theorem}
\usepackage{color}
\usepackage{xcolor}
\usepackage{listings}
\usepackage{framed}
\usepackage{caption,graphicx,newfloat}
\DeclareCaptionFont{white}{\color{white}}
\DeclareCaptionFormat{listing}{\colorbox{gray}{\parbox{\textwidth}{#1#2#3}}}
\lstdefinestyle{framed}
{
     frame=lrb,         
     belowcaptionskip=-1pt,
     xleftmargin=8pt,
     framexleftmargin=8pt,
     framexrightmargin=5pt,
     framextopmargin=5pt,
     framexbottommargin=5pt,
     framesep=0pt,
     rulesep=0pt,
 }
\newcommand{\MS}[1]{\ensuremath{\Delta}}

\newtheorem{corollary}{Corollary}[theorem]

\newcommand{\topk}{{\textsf{Top-$k$}}}
\newcommand{\topt}{{\textsf{Top-$2$}}}
\newcommand{\topm}{{\textsf{Top-$(k-1)$}}}
\newcommand{\tops}{{\textsf{Top-$(k-2)$}}}

\newcommand{\RMQ}{{\textsf{RMQ}}}

\newcommand\ceil[1]{\lceil #1 \rceil}
\newcommand\floor[1]{\lfloor #1 \rfloor}

\newtheorem{proposition}{Proposition}
\begin{document}

\begin{frontmatter}

\title{Encoding two-dimensional range top-$k$ queries\tnoteref{mytitlenote}}
\tnotetext[mytitlenote]{Preliminary version of these results have appeared in the proceedings of the 27th Annual Symposium on Combinatorial Pattern Matching (CPM 16)~\cite{DBLP:conf/cpm/JoLS16} and the 29th International Symposium on Algorithms and Computation (ISAAC 2018)~\cite{DBLP:conf/isaac/JoS18}}

\author[Seungbum]{Seungbum Jo}
\address[Seungbum]{Chungbuk National University}
\ead{sbjo@chungbuk.ac.kr}

\author[Rahul]{Rahul Lingala}
\address[Rahul]{IIT Bombay}
\ead{lingalarahul7@gmail.com}

\author[Satti]{Srinivasa Rao Satti}
\address[Satti]{Seoul National University}
\ead{ssrao@cse.snu.ac.krl}

\begin{abstract}
We consider the problem of encoding two-dimensional arrays, whose elements come from a total order, for answering \topk{} queries. The aim is to obtain encodings that use space close to the information-theoretic lower bound, which can be constructed efficiently. For an $m \times n$ array, with $m \le n$, we first propose an encoding for answering 1-sided \topk{} queries, whose query range is restricted to $[1 \dots m][1 \dots a]$, for $1 \le a \le n$. 
Next, we propose an encoding for answering for the general (4-sided) \topk{} queries that takes
$(m\lg{{(k+1)n \choose n}}+2nm(m-1)+o(n))$ bits, which generalizes the \textit{joint Cartesian tree} of Golin et al. [TCS 2016]. Compared with trivial $O(nm\lg{n})$-bit encoding, our encoding takes less space when 
$m = o(\lg{n})$. In addition to the upper bound results for the encodings, we also give lower bounds on encodings for answering $1$ and $4$-sided \topk{} queries, which show that our upper bound results are almost optimal.  

\end{abstract}

\begin{keyword}
Encoding model \sep top-$k$ query \sep range minimum query.
\end{keyword}

\end{frontmatter}


\section{Introduction}
\label{sec:intro}
Given a one-dimensional (1D) array $A[1 \dots n]$ of $n$ elements from a total order, the \textit{range \topk{} query on $A$} (\topk{}$(i,j, A), 1 \le i,j \le n$)  returns the positions of $k$ largest values in $A[i \dots j]$. In this paper, we refer to these queries as 2-sided \topk{} queries; and the special case where the query range is $[1 \ldots i]$, for $1 \le i \le n$, as the 1-sided \topk{} queries.
We can extend the definition to the two-dimensional (2D) case -- given an $m \times n$ 2D array $A[1 \dots m][1 \dots n]$ of $mn$ elements from a total order and a $k \in \{1, \dots, mn\}$, the \textit{range \topk{} query on $A$} (\topk{}$(i,j,a,b,A), 1 \le i,j \le m , 1 \le a,b \le n$)  returns the positions of $k$ largest values in $A[i \dots j][a \dots b]$. Without loss of generality, we assume that all elements in $A$ are distinct (by ordering equal elements based on the lexicographic order of their positions). Also, we assume that $m \le n$. In this paper, we consider the following types of \topk{} queries.
\begin{itemize}
	\item {Based on the order in which the answers are reported}
	\begin{itemize}
		\item {Sorted query}: the $k$ positions are reported in sorted order of their corresponding values.
		\item {Unsorted query}: the $k$ positions are reported in an arbitrary order.
	\end{itemize}	
	\item {Based on the query range}
	\begin{itemize}
		\item {1-sided query}: the query range is $A[1 \dots m][1 \dots b]$, for $1 \le b \le n$.
		\item {4-sided query}: the query range is $A[i \dots j][a \dots b]$, for $i, j \in \{1,m\}$, and $a, b \in \{1,n\}$.
	\end{itemize}
\end{itemize}

We consider how to support these range \topk{} queries on $A$ in the \textit{encoding model}. 
In this model, one needs to construct a data structure (an encoding) so that queries can be answered by only accessing the data structure, without accessing the original input array $A$. The minimum size of such an encoding is also referred to as the \textit{effective entropy} of the input data~\cite{DBLP:journals/tcs/GolinIKRSS16}.
Our aim is to obtain encodings that use space close to the effective entropy, which can be constructed efficiently. 
In the rest of the paper, we use $\topk{}(i, j, a, b)$ to denote $\topk{}(i, j, a, b, A)$ 
if $A$ is clear from the context. Also, unless otherwise mentioned, we assume 
that all \topk{} queries are sorted, and 4-sided \topk{} queries. Finally, we assume the 
standard word-RAM model~\cite{miltersen-survey} with word size $\Theta(\lg{n})$\footnote{We use $\lg{n}$ to denote $\log_2{n}$.}. 

\subsection{Previous work.}
The problem of encoding 1D arrays to support \topk{} queries has been 
widely studied in the recent years. Especially, the case when $k=1$, which is  
commonly known as the \textit{Range maximum query} (\RMQ{}) problem, has been 
studied extensively, and has a wide range of applications~\cite{DBLP:journals/siamcomp/BerkmanV93}. 
Fischer and Heun~\cite{FischerH10} proposed an optimal $2n+o(n)$-bit data structure which answers \RMQ{} queries on 1D array of size $n$ in constant time.
For a 2D array $A$ of size $m \times n$, a trivial way to encode $A$ for answering
\RMQ{} queries is to store the rank of all elements in $A$, using $O(nm\lg{n})$ bits.
Golin et al.~\cite{DBLP:journals/tcs/GolinIKRSS16} show that when $m=2$ and 
$\RMQ{}$ encodings on each row are given, one can support $\RMQ{}$ queries on $A$ 
using $n-O(\lg{n})$ extra bits by encoding a \textit{joint Cartesian tree} of the two rows.
By extending the above encoding, they obtained $nm(m+3)/2$-bit encoding for answering $\RMQ{}$ queries on $A$, which takes less space than the trivial $O(nm\lg{n})$-bit encoding when $m=o(\lg{n})$.
Brodal et al.~\cite{BDR12} proposed an $O(\min{(nm\lg{n}, m^2n)})$-bit data structure which supports $\RMQ$ queries on $A$ in constant time. 
Finally, Brodal et al.~\cite{DBLP:conf/esa/BrodalBD13} obtained an optimal $O(nm\lg{m})$-bit encoding for answering $\RMQ{}$ queries on $A$ (although the queries are not supported efficiently). 
For the case when $k=2$, Davoodi et al.~\cite{Davoodi20130131} proposed a $3.272n+o(n)$-bit 
data structure to encode a 1D array of size $n$, which supports $\topt{}$ queries in 
constant time. The space was later improved by Gawrychowski and Nicholson~\cite{DBLP:conf/icalp/GawrychowskiN15}
to the optimal $2.755n+o(n)$ bits, although it does not support queries efficiently.

For general $k$, on a 1D array of size $n$, Grossi et al.~\cite{DBLP:journals/talg/GrossiINRR17} proposed 
an $O(n\lg{k})$-bit encoding 
which supports sorted \topk{} queries in $O(k)$ time, and showed that at least $n\lg{k}-O(n)$ bits are necessary for answering
1-sided \topk{} queries; Gawrychowski 
and Nicholson~\cite{DBLP:conf/icalp/GawrychowskiN15} 
proposed a $(k+1)nH(1/(k+1))+o(n)$-bit\footnote{$H(x) = 
	x\lg{(1/x)}+(1-x)\lg{(1/(1-x))}$, i.e., an entropy of the binary string whose density of zero is $x$} 
encoding for \topk{} queries (although the queries are not supported efficiently), 
and showed that at least $(k+1)nH(1/(k+1))(1-o(1))$ bits are
required to encode \topk{} queries. 
They also proposed a  $(k+1.5)nH(1.5/(k+1.5))+o(n\lg{k})$-bit data structure for answering \topk{} queries 
in $O(k^6\lg^2{n}f(n))$ time, for any strictly increasing function $f$.
For a 2D array $A$ of size $m \times n$, one can answer $\topk{}$ queries using $O(nm\lg{n})$ bits, by storing the rank of all elements in $A$.
To support the queries efficiently on this encoding, one can use some of the existing orthogonal range reporting data structures in 3D, in which the $z$-coordinate stores the rank of the elements in $A$ (while $x$- and $y$-coordinates correspond to the positions of the elements in $A$), while reporting the points in sorted order of their ranks. However, all the known 3D orthogonal range reporting data structures use at least linear space (i.e., $O(nm\lg{n})$ bits), and take at least $O(k \cdot polylog(n))$ time to answer sorted $\topk{}$ queries. See~7\cite{DBLP:conf/birthday/He13} for details.
Also Rahul and Tao considered the data structures for answering $\topk{}$ queries in $\mathbb{R}^2$~\cite{DBLP:conf/pods/RahulT15,DBLP:conf/pods/RahulT16}; all their data structures use super-linear space (for $4$-sided queries).
To the best of our knowledge, there are no results on encodings for range $\topk{}$ queries on 2D arrays for $k > 1$. 

\begin{table}
	\centering
	\scalebox{0.8}{
		\begin{tabular}{c | c | c | c}
			\hline
			Dimension & Query type & Space (in bits) & Reference \\
			\hline
			\multicolumn{4}{c}{Upper bounds}\\
			\hline
			
			\multirow{3}{*}{$2 \times n$}&  \multirow{3}{*}{4-sided, sorted} &  $5n-O(\lg{n})$  & ~\cite{DBLP:journals/tcs/GolinIKRSS16}, $k=1$\\
			&   &  $2\lg{{3n \choose n}}+3n+o(n)$  & Theorem~\ref{thm:top2_2n}, $k =2$\\
			&   &  $2\lg{{(k+1)n \choose n}}+4n+o(n)$  & Theorem~\ref{thm:2n}\\
			\hline
			\multirow{6}{*}{$m \times n$} & 1-sided, unsorted & $O(\min{\{nk\lg{(m/k)}, nm\lg{(k/m)}}\})$&Theorem~\ref{thm:one_sided_unsortedds}\\\cline{2-4}
			&1-sided, sorted&$O(\min{\{nk \lg m, nm\lg k\}})$ &Theorem~\ref{thm:one_sided_sortedupper}\\\cline{2-4}
			&\multirow{4}{*}{4-sided, sorted}  &  $O(\min{(nm\lg{n}, m^2n)})$  & ~\cite{DBLP:journals/tcs/GolinIKRSS16}, $k=1$\\ 
			
			&  &  $O(nm\lg{m})$  & ~\cite{DBLP:conf/esa/BrodalBD13}, $k=1$\\
			&  &  $O(nm\lg{n})$  & Trivial\\ 
			&   &  $m\lg{{(k+1)n \choose n}}+2nm(m-1)+o(n)$  & Corollary~\ref{col:mn}\\

			\hline
			\multicolumn{4}{c}{Lower bounds}\\
			\hline
			\multirow{3}{*}{$2 \times n$} &  4-sided &  $5n-O(\lg{n})$  & ~\cite{DBLP:journals/tcs/GolinIKRSS16}, $k=1$ \\
			&  1(or 4)-sided, unsorted &  $1.27(n-k/2)-o(n)$  & (*) Theorem~\ref{lower:unsorted}\\
			&  4-sided, sorted &  $2n-O(\lg{n})$  & (*) Theorem~\ref{lower:sorted}\\
			\hline
			\multirow{3}{*}{$m \times n$} &1-sided, unsorted&Same as the upper bound &Theorem~\ref{thm:one_sided_unsortedds}\\
			&1-sided, sorted& Same as the upper bound &Theorem~\ref{thm:one_sided_sortedupper}\\
			&  4-sided, sorted &  $\Omega(nm\lg{(\max{}(m,k))})$  & ~\cite{BDR12,DBLP:journals/talg/GrossiINRR17}\\
			\hline
	\end{tabular}}
	\caption{Summary of the results of upper and lower bounds for \topk{} encodings on 2D arrays. The lower bound results marked (*) (of Theorem~\ref{lower:unsorted} and \ref{lower:sorted}) are for the additional space (in bits) necessary, assuming that encodings of $\topk{}$ queries for both rows are given. Results in \cite{DBLP:journals/tcs/GolinIKRSS16} support $O(1)$-time queries using $o(n)$ extra bits.} 
	\label{tab:summary}
\end{table}

\subsection{Our results.}
Given an $m \times n$ array $A$, we first show that $O(\min{\{nm \lg k, nk\lg m\}})$ bits are necessary and sufficient  for answering sorted 1-sided \topk{} queries.
For unsorted 1-sided queries, 
 we show that $O(\min{\{nk\lg{(m/k)}, nm\lg{(k/m)}}\})$ bits are necessary and sufficient. 
This space bound is strictly less than the 
space used to encode sorted  1-sided \topk{} queries. Thus it is interesting to note that in 2D, 
there is a gap between the space needed to encode the 1-sided \topk{} queries for the 
sorted and unsorted cases. In contrast, in 1D, the space needed to encode the 1-sided
\topk{} queries for both sorted and unsorted cases are asymptotically the same (for $k = o(n)$).
We show that such encodings can be simply constructed in $O(mn\lg{k}+nk)$ time by using a min-heap data structure.

Next, in Section~\ref{sec:4sided}, we consider encodings for 4-sided $\topk{}$ queries on an $m \times n$ array $A$.
We first observe that one can obtain an $O(mn\lg{n})$-bit data structure 
which answers 4-sided \topk{} queries on $A$ in $O(k)$ time, 
by combining the results of~\cite{BDR12} and~\cite{DBLP:conf/isaac/BrodalFGL09}. 
We then propose the alternative encoding which uses less space than the trivial $mn\ceil{\lg{(mn)}}$-bit encoding (which stores all the positions of $A$ in sorted order), for small $m$. 

To be more precise, we first show that $4n$ bits are sufficient for answering sorted 4-sided $\topk{}$ queries on $2 \times n$ array, when encodings for answering sorted 2-sided $\topk{}$ queries for each row are given. This encoding is obtained by a binary DAG for answering $\topt{}$ queries on $2 \times n$ array which  generalizes the $(5n-O(\lg{n}))$-bit encoding of $\RMQ{}$ query on $2 \times n$ array proposed by Golin et al.~\cite{DBLP:journals/tcs/GolinIKRSS16}, to general $k$; and shows that we can encode a joint Cartesian tree for general $k$ (which corresponds to the DAG in our paper) using $4n$ bits. Note that the additional space is independent of $k$. By extending the encoding on $2 \times n$ array, we obtain $(m\lg{{(k+1)n \choose n}}+2nm(m-1)+o(n))$-bit encoding for answering 4-sided \topk{} queries on $m \times n$ arrays. This improves upon the trivial $\ceil{mn\log{(mn)}}$-bit encoding when $m = o(\lg{n})$, and also generalizes the $mn(m+3)/2$-bit encoding~\cite{DBLP:journals/tcs/GolinIKRSS16} for answering $\RMQ{}$ queries.
The trivial encoding of the input array takes $O(nm\lg{n})$ bits, whereas one can easily show a lower bound of $\Omega(nm \lg{(\max{(m,k)})})$ bits for any encoding of an $m \times n$ array that supports \topk{} queries since at least $O(nm\lg{m})$ bits are necessary for answering $\RMQ{}$ queries~\cite{BDR12}, and at least $n\lg{k}$ bits are necessary for answering $\topk{}$ queries for each row~\cite{DBLP:journals/talg/GrossiINRR17}. Thus, there is only a small range of parameters where a strict improvement over the trivial encoding is possible. Our result closes this gap partially, achieving a strict improvement when $m = o(\lg n)$.

In Section~\ref{sec:2nds}, we also obtain a data structure for answering $\topk{}$ queries in $O(k^2+kT(n,k))$ time using $2S(n,k)+(4k+7)n+ko(n)$ bits, if there exists an $S(n, k)$-bit encoding to answer sorted 2-sided $\topk{}$ queries on a 1D array of size $n$ in $T(n,k)$ time. 
Comparing to the $2S(n,k)+4n+o(n)$-bit encoding, 
this data structure uses more space but supports \topk{} queries efficiently (the $2S(n,k)+4n+o(n)$-bit encoding takes $O(k^2n^2+nkT(n,k))$ time for answering $\topk{}$ queries).

Finally, in Section~\ref{sec:lower}, 
given a $2 \times n$ array $A$, 
we consider the lower bound on additional space required
to answer unsorted (or sorted) $\topk{}$ on $A$ when encodings of \topk{} query for each row are given. 
We show that at least $1.27(n-k/2)-o(n)$ additional bits are necessary for answering unsorted 1-sided \topk{} queries on $A$, when encodings of unsorted 1 or 2-sided \topk{} query for each row are given. Note that this lower bound also gives the lower bound for answering 
unsorted 4-sided $\topk{}$ queries on $2 \times n$ array under the same condition. 
We also show that $2n-O(\lg{n})$ additional bits are necessary for answering sorted 4-sided \topk{} queries on $A$, when encodings of unsorted (or sorted) 2-sided \topk{} query for each row are given.
These lower bound results imply that our encodings in \ref{sec:4sided} are close to optimal (i.e., within $O(n)$ bits of the lower bound), since any \topk{} encoding for the array $A$ also needs to support the \topk{} queries on the individual rows.
All these results are summarized in Table~\ref{tab:summary}.
\section{Encoding 1-sided range \topk{} queries on two dimensional array}
\label{sec:1sided}
In this section, we consider the encoding of sorted and unsorted 1-sided \topk{} queries on a 2D array $A[1 \dots m][1 \dots n]$.
In the rest of the paper, we use $(i,j)$ to denote the position in the $i$-th row and $j$-th column of a 2D array.
\subsection{Sorted 1-sided queries}

We first introduce an encoding by simply extending the encoding of
sorted 1-sided \topk{} queries for 1D array
proposed by Grossi et al.~\cite{DBLP:conf/esa/GrossiINRR13}.
Next we propose an optimal encoding for sorted 1-sided \topk{} queries on $A$.
For a 1D array $A'[1 \dots n]$, define another 1D array $X[1 \dots n]$ as follows.
For $1 \le i \le k$, define $X[i] = i$. For $k < i \le n$, $X[i] = X[i']$
if there exists a position $i'$ which satisfies $\topk{}(1, i, A') = \topk{}(1, i-1, A) \setminus \{i'\} \cup \{i\}$, and $X[i] = k+1$ otherwise.

Then one can answer the \topk{}$(1, i, A')$
by finding the rightmost occurrence of every element $1 \dots k$ in $X[1 \dots i]$.
By representing $X$ (along with some additional auxiliary structures) using $n\lg{k}+O(n)$ bits,
Grossi et al.~\cite{DBLP:conf/esa/GrossiINRR13} obtained an encoding which
supports 1-sided \topk{} queries on $A'$ in $O(k)$ time.
For a 2D array $A$, one can encode $A$ to support sorted 1-sided \topk{} queries
by writing down the values of $A$ in column-major order into a 1D array,
and using the encoding described above -- resulting in the following encoding.

\begin{proposition}
\label{prop;one_sided}
A 2D array $A[1 \dots m][1 \dots n]$ can be encoded using
$mn\lg{k} + O(n)$ bits to support sorted 1-sided \topk{} queries in $O(k)$ time. 
\end{proposition}

Now we describe an optimal encoding of $A$ which supports sorted 1-sided \topk{} queries.
For 1D array $A'[1 \dots n]$, we can define another 1D array $B'[1 \dots n]$ such that
for $1 \le i \le n$, $B'[i] = l$ if $A'[i]$ is the $l$-th largest element in $A'[1 \dots i]$ with $l \le k$, and  $B'[i] = k+1$ otherwise.
Then we answer the \topk{}$(1, i, A')$ query as follows. We first find the rightmost position $p_1 \le i$ such that $B'[p_1] \le k$.
Then we find the rest of $k-1$ positions $p_{k} < p_{k-1} < \dots < p_2$ such that for $2 \le j \le k$, $p_j$ is the rightmost position in 
$A'[1 \dots p_{j-1}-1]$ with $B'[p_j] \le k-j+1$. Finally, we return the positions $p_1, p_2, \dots, p_k$. 
Therefore by storing $B'$ using $n\ceil{\lg{(k+1)}}$ bits, we can answer the sorted 1-sided \topk{} queries on $A'$.
Also we can sort $A'[p_1], \dots, A'[p_k]$ using the property that for $1 \le b <  a \le k$, 
$A'[p_a] < A'[p_b]$ if and only if one of the following two conditions hold: (i) $B'[p_a] \ge B'[p_b]$, or 
(ii) $B'[p_a] < B'[p_b]$ and there exist at least $q = B'[p_b] - B'[p_a]$
positions $j_1, j_2, \dots, j_q$ between $p_a$ and $p_b$ where $B'[j_r] \le B'[p_a]$ for all $1 \le r \le q$.

We can extend this encoding for the sorted 1-sided \topk{} queries on a 2D array~$A$, to obtain an optimal encoding as stated in the following theorem.
\begin{theorem}
\label{thm:one_sided_sortedupper}
Given a 2D array $A[1 \dots m][1 \dots n]$, there is an encoding of size
$O(\min{\{nm \lg k, nk\lg m\}})$ bits, which can answer
sorted 1-sided \topk{} queries.
Also, the space bound is asymptotically optimum.
\end{theorem} 
\begin{proof}
For $1 \le j \le n$, we first define the elements of $j$-th column in $A$ as $a_{1j} \dots a_{mj}$.
Then we define the sequence $S_j = s_{1j} \dots s_{mj}$ such that for $1 \le i \le m$, $s_{ij}= l$ if $a_{ij}$
is the $l$-th largest element in $A[1 \dots m][1 \dots j]$ with $l \le k$,
and $s_{ij} = k+1$ otherwise. 
Since there exist $T = \sum_{p=0}^{\min{(m,k)}}{m \choose p}{k \choose p}p!$
possible $S_j$ sequences
($T$ is the total number of ways in which we can choose $p$ out of the $m$ rows for new entries 
into the \topk{} positions, summed over all possible values of $p$),
we can store $S^{A} = S_1 \dots S_n$ using $n\ceil{\lg{T}} = O(\min{\{nm \lg k, nk\lg m\}})$
bits and we can answer the sorted 1-sided \topk{}$(1, m, 1, j)$ queries on $A$ by the following procedure.
\begin{enumerate}
\item Find the rightmost column $q$, for some $q \le j$, such that $S_{q}$ has $\ell >0$ elements
$s_{p_{1}q},  \dots, s_{p_{\ell}q}$ where  $s_{p_{1}q} <  \dots < s_{p_{\ell}q} < k+1$. 
If $\ell = k$, we return the positions of $(p_1, q) \dots (p_k, q)$ as the answers of the query, and stop.
Otherwise (if $\ell < k$), we return the positions of $(p_1, q) \dots (p_{\ell}, q)$, and 
\item  Repeat Step 1 by setting $k$ to $k-\ell$, and $j$ to $q-1$.
\end{enumerate}
We can return the positions in the sorted order of their corresponding values by applying the procedure to sort the answers of 1-sided $\topk{}$ queries on $S^A$, described above. This shows the upper bound.
%

We now show that
$n\lg{T}$ bits are necessary to encode sorted 1-sided \topk{} queries on $A$. 
Suppose there are two distinct sequences $S^{A} = S_1 \dots S_i$ and $S^{A'}=S'_1 \dots S'_i$
which give sorted 1-sided \topk{} encodings of 2D arrays $A$ and $A'$, respectively. 
For $1 \le b \le n$, if $S_b \neq S'_b$ then \topk{}$(1, m, 1, b, A)$ $\neq$ \topk{}$(1, m, 1, b, A')$
by the definition of $S^{A}$ and $S^{A'}$.
Since for an $m \times n$ array, there are $T^n$ distinct sequences $S^{A_1} \dots S^{A_{T^{n}}}$,
We now prove for $1 \le q \le T^n$,
each $S^{A_q} = S^q_1 \dots S^q_n$ has an array $A$ such that $S^{A} = S^{A_q}$, which completes the proof of the theorem.

Without loss of generality, suppose that all elements in $A$ come from the set $L = \{1, \dots , mn \}$.
Then we can reconstruct $A$ from the rightmost column using $S^{A_q}$ as follows.
If $s^q_{jn}  \le k $, for $1 \le j \le m$, we assign the $s^q_{jn}$-th largest element in $L$ to $A[j][n]$.
After we assign all values in the rightmost column with $s^q_{jn} \le k $,
we discard all assigned values from $L$, move to $(n-1)$-th column and repeat the procedure.
After we assign all values in $A$ whose corresponding values in $S^{A_q}$ are smaller than $k+1$,
we assign the remaining values in $L$ to remaining positions in $A_q$ which are not assigned yet.
Thus for any $1 \le b \le n$, if  $S_b^{q}$ has $\ell >0$ elements
$s_{p_{1}b},  \dots, s_{p_{\ell}b}$ where  $s_{p_{1}b} <  \dots < s_{p_{\ell}b} < k+1$, 
then the $b$-th column in $A$ contains $\ell$-largest elements in $A[1 \dots m][1 \dots b]$ by the above procedure.
This shows that $S^A = S^{A_q}$.
\qed
\end{proof}
Note that this encoding takes less space than the encoding in the Proposition~\ref{prop;one_sided}.

\subsection{Unsorted 1-sided queries}
In this section, we consider the encoding of unsorted 1-sided $\topk{}(1, m, 1, b)$ queries on a 2D array $A$.
Let unsorted $\topk{}(1, m , 1, b) = b_1, \dots, b_k$ be the $k$ positions, ordered in their lexicographic order.
Grossi et al~\cite{DBLP:conf/esa/GrossiINRR13} show that for any array $B[1 \dots n]$ of size $n$,
one can support unsorted $\topk{}(1, i, B)$ queries for $1 \le i \le n$ using $n\lg{k}+O(n))$ bits with $O(k)$ query time. 
This implies that we can answer unsorted 1-sided \topk{} queries on $A$ using $mn\lg{k}+O(mn)$ bits with $O(k)$ query time, by converting the unsorted 1-sided \topk{} queries on $A$ into the unsorted \topk{} queries on a 1D array of size $mn$ which is obtained by the values of $A$ in column-major order. Now we consider the another encoding which supports unsorted query using optimal space, if query time is not of concern, and show the following.

\begin{theorem}
\label{thm:one_sided_unsortedds}
Given a 2D array $A[1 \dots m][1 \dots n]$, there is a data structure of size $O(\min{\{nk\lg{(m/k)}, nm\lg{(k/m)}}\})$ bits
which supports unsorted 1-sided \topk{} queries.
Moreover, the space bound is asymptotically optimal.
\end{theorem}
\begin{proof}
We first show the upper bound by considering the two cases, (i) $k \le m$, and (ii) $k > m$ separately.
\\\\
\noindent
{\bf Case (i) $k \le m$. } 
In this case, we first encode the answers of unsorted $\topk{}(1, m, 1, 1)$ query using $\lceil \lg{{m} \choose {k}} \rceil$ bits. 
For $b > 1$, the answers of $\topk{}(1, m , 1, b)$ are the positions of $k$ largest values in $A$ out of $m+k$ positions, corresponding to the positions in $\topk{}(1, m , 1, b-1)$ along with the positions $\{(1, b), \dots, (m, b)\}$. 
Thus, for any $1 < b \le n$, if we know the answers of unsorted $\topk{}(1, m , 1, b-1)$, we can encode the answers of unsorted $\topk{}(1, m , 1, b)$ using $\lceil \lg{{k+m} \choose {k}} \rceil$ bits. Hence, the total space for encoding the answers for
all the $n$ columns is $\lceil \lg{{m} \choose {k}} \rceil + (n-1)\lceil \lg{{k+m} \choose {k}} \rceil$ bits. 
We can find the answers of a 1-sided $\topk{}(1, m , 1, b)$ query as follows. We first decode the answers of $\topk{}(1, m , 1, 1)$, and decode the answers of unsorted 1-sided query from left to right until we decode the answers of $b$-th unsorted 1-sided \topk{} query.
\\
\noindent
{\bf Case (ii) $k > m$. }
When $1 \le b\le \lfloor k/m \rfloor$, it is obvious that $\topk{}(1, m, 1, b) = \{(i, j) | 1 \le i \le m, 1 \le j \le b\}$, i.e., all the positions within the query range are part of the answer. Therefore no extra space is needed for storing the answers of unsorted $\topk{}(1, m, 1, b)$ queries for $1 \le b\le \lfloor k/m \rfloor$. 
When $b> \lfloor k/m \rfloor$, we can encode the answers of $\topk{}(1, m, 1, b)$ queries using $(n-1-\lfloor k/m \rfloor)\lceil \lg{{k+m} \choose {k}}\rceil + \lceil\lg{{m} \choose {k-m\lfloor k/m \rfloor}}\rceil$ bits by using the similar encoding as described above (using $\lceil\lg{{m} \choose {k-m\lfloor k/m \rfloor}}\rceil$ bits for answering $\topk{}(1, m , 1, \floor{k/m}+1)$ query and $(n-1-\lfloor k/m \rfloor)\lceil \lg{{k+m} \choose {k}}\rceil$ bits for answering $\topk{}(1, m , 1, b)$ queries for any $b >\floor{k/m}+1$). 
Also, we can find the answers of unsorted 1-sided $\topk{}(1, m , 1, b)$ by a similar procedure as in Case (i). The only difference is when $1 \le b\le \lfloor k/m \rfloor$, we just report all positions in the sub-array $A[1 \dots m][1 \dots b]$.  
\\\\
From the above two cases, it follows that we can answer unosorted 1-sided \topk{} queries on $A$ using at most $n \log {{k+m} \choose k}$ bits.
If $k < m$, then $n \log {{k+m} \choose k} \approx nk \log (m/k)$; and if $k > m$, then $n \log {{k+m} \choose k} \approx nm \log (k/m)$. Thus the space bound can be written as
$O(\min{\{nk\lg{(m/k)}, nm\lg{(k/m)}}\})$ bits. 
This shows the upper bound stated in the theorem.

Now we show the lower bound.
Without loss of generality, suppose that all elements in the array come from the 
set $L_n = \{1, . . . ,mn\}$, and $k \le m$ (we can prove the case when $k > m$ 
in a similar way). Then it is enough to show that there are 
$N = {{m} \choose {k}} {{k+m} \choose {k}}^{(n-1)}$ arrays $A_1 \dots A_N$ 
of size $m \times n$ such that for any $1 \le s \neq t \le N$, there exists $1 \le b \le n$ 
such that unsorted $\topk{}(1, m, 1, b, A_s) \neq \topk{}(1, m, 1, b, A_t)$. 
 
We prove the above statement by induction on $n$. When $n=1$, we assign $\{m, 
\dots m-k+1 \}$ to the answers of unsorted $\topk{}(1,m,1,1)$, and assign 
$L_1-\{m, \dots m-k+1 \}$ to remaining positions arbitrary. 
Since any $k$ positions in $(i, 1), 1 \le i \le m$ can be the answer of 
$\topk{}(1, m, 1, 1)$ query, there are ${m} \choose {k}$ arrays such that any 
two of them have different answers for the $\topk{}(1, m, 1, 1)$ query. 

Now assume the inductive hypothesis that the above statement holds for some $n'$ 
where $1 \le n' < n$, and that there 
are $N' = {{m} \choose {k}} {{k+m} \choose {k}}^{(n'-1)}$ arrays $A_1 \dots A_{N'}$,
satisfying the above statement. Let a set $M=L_{(n'+1)}-L_{n'} = \{m(n'+1), 
\dots m(n'+1)-k+1\}$ and $M' = \emptyset$.

To prove the inductive step for $n'+1$, we first pick an arbitrary $A_s$, $1 \le s \le N'$ 
from the set of arrays $\{A_1 \dots A_{N'}\}$ and add the $(n'+1)$-th column 
to $A_s$ (none of the positions in this column have an assigned value).
Then for some integer $\alpha$ where $0 \le \alpha \le k$, we pick the answers of 
unsorted $\topk{}(1, m, 1, n'+1)$ query by choosing $\alpha$ positions from the set $\topk{}(1, m, 1, n', A_s) = \{n'_1 
\dots n'_k\}$ and choosing the remaining $k-\alpha$ positions from the $(n'+1)$-th column.
We then choose $\alpha$ elements from $M$ and add them to the $\alpha$ positions chosen from 
the set $\{n'_1 \dots n'_k\}$.
When we assign a value $x \in M $ to $(i, j) 
\in \{n'_1 \dots n'_k\}$, we delete $x$ from $M$, set $M' = M' \cup A_s[i][j]$ and $A_s[i][j] = x$. 
Since $n'_i \le n'm$ for all $1 \le i \le k$, it is easy to show that this does 
not change the answers of unsorted  $\topk{}(1, m, 1, b, A_s)$ for all $1 \le b 
\le n'$. 
Next, we assign the remaining values in $M$ to the $k -\alpha$ chosen positions in the
$(n'+1)$-th column, and finally assign the values in $M'$ to the remaining positions in that
column arbitrarily.
%
Since there are $\sum_{\alpha=1}^{k} {{k} \choose {\alpha}}{{m} \choose {k-\alpha}}={{m+k} \choose {k}}$ ways to select the answers of unsorted $\topk{}(1, m, 1, n'+1)$ query, and for each $s$, there are $Q = {{m+k} \choose {k}}$ arrays $A_{s1} \dots A_{sQ}$ such that for all $1 \le t \neq \ell \le Q$ and $1 \le b \le n'$, unsorted $\topk{}(1, m, 1, b, A_{st})=\topk{}(1, m, 1, b, A_{s\ell}) = \topk{}(1, m, 1, b, A_s)$ but unsorted $\topk{}(1, m, 1, n'+1, A_{st}) \neq \topk{}(1, m, 1, n'+1, A_{s\ell})$. 
Therefore the above statement holds for $n'+1$ whenever it holds for $n'$, which proves the theorem.
\qed
\end{proof}

Note that we can construct the encoding of Theorem~\ref{thm:one_sided_unsortedds} in $O(mn\lg{k}+nk)$ time by maintaining a min-heap of size at most $k$. More precisely, we insert the values of $A$ in column-major order and delete the minimum value in the heap when the size of the heap is more than $k$.
We can answer the position of the $k$ largest values in $A[1,m][1,i]$ for $1 \le i \le n$, by scanning the heap after every $m$-th insertion, in $O(k)$ time.
\\\\
\noindent\textbf{Remark}
Let $S_{sorted}$ and $S_{unsorted}$ be the space needed to encode the sorted and unsorted 1-sided \topk{} queries respectively. 
For 1D array, Gawrychowski and Nicholson showed that $S_{sorted}/S_{unsorted} \le 2$ (thus, the space requirements are asymptotically same for both case)~\cite{DBLP:conf/cpm/GawrychowskiN15}. 
In contrast, in 2D array case when $k = m$, $S_{sorted}/S_{unsorted} \le \lg k$
by Theorems~\ref{thm:one_sided_sortedupper} and \ref{thm:one_sided_unsortedds}, which implies the gap between the space needed to encode the 1-sided \topk{} queries for sorted and unsorted case for a 2D array is significantly more than the case for a 1D array.
\section{Encoding 4-sided \topk{} queries on on $2 \times n$ array}
\label{sec:4sided}
In this section, we give an encoding which supports general \topk{} queries on $m \times n$ 2D array $A$. Note that if query time is not of concern in the above data structure, one can simply consider the $mn\ceil{\lg{(mn)}}$-bit trivial encoding for answering $\topk{}$ queries on $A$, by storing the rank (i.e., the position in sorted order) of all the elements in $A$.
We first introduce an $O(mn\lg{n})$-bit data structure which supports  \topk{} query in $O(k)$ time
by using the \RMQ{} encoding of Brodal et al.~\cite{DBLP:conf/esa/BrodalBD13}.

\begin{proposition}
\label{prop;four_sided}
Given a 2D array $A[1 \dots m][1 \dots n]$, 
there exists an $O(mn\lg{n})$-bit data structure
to support unsorted \topk{}$(i, j, a, b, A)$ in $O(k)$ time
for $1 \le a,b \le m $ and $1 \le i,j \le n$.
\end{proposition}  
\begin{proof}
We use a data structure similar to the one outlined in~\cite{DBLP:conf/isaac/BrodalFGL09}
(based on Frederikson's heap selection algorithm~\cite{DBLP:journals/iandc/Frederickson93}) 
for answering unsorted \topk{} queries in 1D array\footnote{Brodal et al.~\cite{DBLP:conf/isaac/BrodalFGL09}  
also give another structure to answer sorted \topk{} queries, with the same time and space bounds.}. 
First encode $A$ using $O(mn\lg{n})$ bits to support \RMQ{} (range maximum) queries in constant time for any rectangular range in $A$.
This encoding also supports finding the rank of any element in $A$ in $O(1)$ time~\cite{BDR12}.
Next, let $x = A[x_1][x_2]$ be the maximum value in $A[i \dots j][a \dots b]$, which can be found using an \RMQ{} query on $A$. 
Then consider the $4$-ary heap obtained by the following procedure.
The root of the heap is $x$, and its four subtrees are formed by recursively constructing the $4$-ary heap
on the sub-arrays $A[i \dots x_1-1][a \dots b]$, $A[x_1+1 \dots j][a \dots b]$, $A[x_1][a \dots x_2-1]$ and $A[x_1][x_2+1 \dots b]$,
respectively.
Now, we can find the $k$ largest elements in the above $4$-ary heap in $O(k)$ time using the algorithm proposed by 
Frederickson~\cite{DBLP:journals/iandc/Frederickson93} (note that this algorithm only builds a heap with $O(k)$ nodes
which is a connected subgraph of the above $4$-ary heap).
\qed
\end{proof}
We now introduce alternative encoding to support \topk{} queries on an $m \times n$ 2D array $A$, which take less space than the trivial encoding for small $m$. The overall idea is as follows. 
We first show that $4n$ bits are sufficient for answering sorted 4-sided $\topk{}$ queries on $A$ for $m=2$ when encodings for answering sorted 2-sided $\topk{}$ queries for each row are given. This encoding is obtained by encoding the binary DAG which is defined later.
After that, we extend the encoding into $m \times n$ array and obtain $(m\lg{{(k+1)n \choose n}}+2nm(m-1)+o(n))$-bit encoding for answering $\topk{}$ queries on $A$. 
Note that this encoding takes less space than
the trivial $mn\ceil{\lg{(mn)}}$-bit encoding when $m = o(\lg {n})$. 

Now we describe the encoding of sorted 4-sided \topk{} on $A$ when $m = 2$. For $1 \le i \le 2$, let $A_i[1 \dots n]$ be the array of the $i$-th row in $A$, and assume that sorted 2-sided $\topk{}$ encodings on $A_1$ and $A_2$ are already given.
When $k=1$, i.e., to answer $\RMQ{}$ queries on $A$, one can use the \textit{joint Cartesian tree} of Golin et al.~\cite{DBLP:journals/tcs/GolinIKRSS16}. 
The joint Cartesian tree of constructs a conceptual binary tree analogous to a Cartesian tree, storing a bit indicating which row the maximum element in the range comes from, splitting the range corresponding to the node at the position of the maximum element, and then recursing on each subrange. Thus by storing an extra $n$ bits (one at each node of the conceptual binary tree), they showed that one can answer $\RMQ{}(1, 2, a, b, A)$ queries for any $1 \le a, b \le n$, if the encodings for answering $\RMQ{}$ on $A_1$ and $A_2$ are given. See~\cite{DBLP:journals/tcs/GolinIKRSS16} for details.
To answer the sorted 4-sided \topk{} queries with $k \ge 2$, we extend the idea of a joint Cartesian tree into a DAG-based structure, denoted by $D^k_A$, which is defined as follows. 
 
Every node $p$ in $D^k_A$ is labeled with some closed interval $I_p = [a, b]$, where $1 \le a, b \le n$. 
In this case, we use both $\topk{}(p)$ and $\topk{}(I_p)$ to refer to the sorted $\topk{}(1, 2, a, b, A)$ query. 
For a node $p$ with label $I_p = [a, b]$ in $D^k_A$ and $1 \le i \le k$, let $(p^{i}_r, p^{i}_c)$ be the 
position of the $i$-th largest element in $A[1, 2][a \dots b]$.
Now we define $D^k_A$ as follows (see Figure~\ref{fig:dag} for an example.). 

\begin{figure}
	\centering
	\scalebox{0.9}{
		\begin{tabular}{|c||c|c|c|c|c|c|c|c|c|}
			\hline
			$A_1$&1&21&17&12&20&3&15&11&10 \\\hline
			$A_2$&6&5&16&14&19&2&18&4&7 \\
			\hline
			
		\end{tabular}
	}
	\vspace{0.02in}
	\usetikzlibrary{positioning}
	\scalebox{0.9}{
		\begin{tikzpicture}[
		node/.style={
			rounded corners=1mm,
			minimum size=3mm,
			thick,
			draw,
		},
		node distance=15mm
		]
		\node(a)[node]{[1,9]};
		\node(b)[node,below left=of a, xshift=-0.7cm]{[1,4]};
		\node(c)[node,below right=of a, xshift=0.7cm]{[3,9]};
		\node(d)[node,below left=of b, xshift=0.7cm]{[1,2]};
		\node(e)[node,below right=of b, xshift=-0.7cm]{[3,4]};
		\node(f)[node,below left =of c, xshift=0.7cm]{[3,6]};
		\node(g)[node,below right =of c, xshift=-0.7cm]{[6,9]};
		
		\node(h)[node,below left=of d, xshift=1.5cm]{[1,1]};
		\node(i)[node,below right=of d, xshift=-1.3cm]{[2,2]};
		
		\node(j)[node,below left=of e, xshift=1.3cm]{[3,3]};
		\node(k)[node,below right=of e, xshift=-1.5cm]{[4,4]};
		
		\node(l)[node,below=of f, xshift=-0.8cm, yshift=0.5cm]{[4,6]};
		\node(m)[node,below left=of g, xshift=1cm]{[6,7]};
		\node(n)[node,below right=of g, xshift=-1cm]{[8,9]};
		
		\node(p)[node,below right=of m, xshift=-1.5cm]{[7,7]};
		
		\node(q)[node,below left=of n, xshift=1.5cm]{[8,8]};
		\node(r)[node,below right=of n, xshift=-1.3cm]{[9,9]};
		
		\node(t)[node,below =of l, yshift=1.25cm]{[5,6]};
		
		\node(u)[node,below left=of t, xshift=1cm, yshift=0.75cm]{[5,5]};
		\node(v)[node,below right=of t, xshift=-0.5cm,yshift=0.75cm]{[6,6]};
		
		\path(a) edge[->] (b);
		\path(a) edge[->] (c);
		\path(b) edge[->] (d);
		\path(b) edge[->] (e);
		\path(c) edge[->] (f);
		\path(c) edge[->] (g);
		\path(f) edge[->] (e);
		
		\path(d) edge[->] (h);
		\path(d) edge[->] (i);
		\path(e) edge[->] (j);
		\path(e) edge[->] (k);
		
		\path(f) edge[->] (l);
		\path(l) edge[->] (k);
		\path(l) edge[->] (t);
		
		\path(g) edge[->] (m);
		\path(g) edge[->] (n);
		
		\path(m) edge[->] (v);
		\path(m) edge[->] (p);
		\path(n) edge[->] (q);
		\path(n) edge[->] (r);
		
		\path(t) edge[->] (u);
		\path(t) edge[->] (v);
		\end{tikzpicture}
	}
	\caption{$2 \times n$ array $A$ and the DAG $D^3_A$.}   
	\label{fig:dag}
\end{figure}

\begin{enumerate}
	\item{} The root of $D^k_A$ is labeled with the range $[1, n]$.
	\item{} A node $p$ with label $I_p = [a,b]$ does not have any child node (i.e., leaf node) if $2(b-a+1) \le k$.
	\item{} Suppose there exists a non-leaf node $p$ with label $I_p = [a, b]$ in $D^k_A$, and let
	$a'$ and $b'$ ($a \le a' \le b' \le b$) be the leftmost and rightmost column indices among the answers of $\topk{}(p)$, respectively.
	If $a < b'$, then the node $p$ has a node with label $[a, b'-1]$ as a left child. Similarly, if $a' < b$, the node $p$ has a node with label $[a'+1, b]$ as a right child.
\end{enumerate}

The following lemma states some useful properties of $D^k_A$.

\begin{lemma}
	\label{lem:dag}
	Let $A$ be a $2 \times n$ array. the following statements hold.
	\begin{itemize}
		\item{(i)} For any two distinct nodes $p$ and $q$ in $D^k_A$, $\topk{}(p) \neq \topk{}(q)$ (i.e., any two distinct nodes have different \topk{} answers). 
		\item{(ii)} $I_p \subset I_q$ if and only if $p$ is descendant of $q$ in $D^k_A$.
		\item{(iii)} For any interval $[a, b]$ with $1 \le a \le b \le n$, there exists 
		a unique node $p$ in $D^k_A$ which satisfies (i) $[a, b] \subset I_p$, and (ii) label of any descendant of $p$ does not contain $[a,b]$.
		Furthermore, for such a node $p$, $\topk{}([a,b]) = \topk{}({p})$.
	\end{itemize}
\end{lemma}
\begin{proof}
(i) From the construction of $D^k_A$, one can observe that if there is a node with label $[a,b]$ in $D^k_A$, 
with $1 < a \le b < n$, then both $(a-1)$-th and $(b+1)$-th column contain at least one element 
that is larger than the elements in $\topk{}([a,b])$, which implies $\topk([a,b]) \neq \topk([a,b+1])$ and $\topk([a-1,b]) \neq \topk([a,b])$.
Now suppose that there are two distinct nodes with labels $[a,b]$ and $[a', b']$ with $b < b'$ such that $\topk([a,b]) = \topk([a', b'])$,
then $\topk([a,b+1]) = \topk([a', b'])$, contradicting the fact that $\topk([a,b]) \neq \topk([a,b+1])$. 
The case when $b > b'$, $a > a'$ or $a < a'$ is analogous. 
\\\\
(ii) Let $I_p = [a_p, b_p]$ and  $I_q = [a_q, b_q]$. From the construction of $D^k_A$, it is the case that if $p$ is a descendant of $q$, then $I_p \subset I_q$.
Now, suppose that there are two nodes $p, q \in D^k_A$ such that $I_p \subset I_q$ but $p$ is not descendant of $q$. 
Then there exists a node $q'$ which  satisfies (i) $q'$ is a descendant of $q$, (ii) $I_p \subset I_{q'}$, and
(iii) no child of $q'$ whose label contains $I_p$.
Since neither of labels of the children of $q'$ contain $I_p$, all the positions of $\topk{}(q')$ are between $a_p$-th and $b_p$-th column.
(otherwise, there always exists a child $q''$ of $q'$ which satisfies $I_{q''} \subset I_p$).
But this would imply that $\topk(q') = \topk(p)$, which leads to a contradiction with Lemma~\ref{lem:dag}(i).
\\\\
(iii) We first show that there exists a unique node $p$ in $D^k_A$ such that $I_p$ contains the interval $[a,b]$ 
and none of labels of the children of $p$ contain $[a,b]$. We then show that the $\topk(p) = \topk([a,b])$.

Since label of the root in $D^k_A$ contains all column indices in $A$, it is easy to see that there exists 
at least one node $p$ with label $I_p = [a_p, b_p]$ in $D^k_A$ such that 
$[a, b] \subset I_p$ but no child of $I_p$ contains $[a,b]$. 
Suppose that there exists another node $p'$ with label  $I_{p'} = [a_p', b_p']$ in $D^k_A$ such that 
$[a, b] \subset I_{p'}$ but there is no child of $p'$ whose label contains $[a,b]$.
By Lemma~\ref{lem:dag}(ii), it follows that $I_p \not\subset I_{p'}$ and $I_{p'} \not\subset I_p$
(otherwise, one of them would be a descendant of the other, contradicting the conditions on $p$ and $p'$).
Now, suppose that $a_p < a_p' < b_p < b_p'$ (the case when $a_p' < a_p < b_p' < b_p$ is analogous).
Then there exists a column $c < a_p'$ such that $p$ has a child node with label $[c, b_p]$ where $[a,b] \subset [c, b_p]$ by the property of $D^k_A$ (note that $a_p' \le a \le b \le b_p$), 
contradicting the fact that $p$ does not have such a child.
This shows that there is a unique such $p$ in $D^k_A$.

Now we claim that $\topk([a, b]) = \topk(p)$. Suppose that there exist a $c \notin [a, b]$ in $I_p$ such that 
column $c$ contains at least one of the answers to $\topk(p)$. 
Also without loss of generality, we assume that $c < a $ (the case when $c > b$ can be handled in a similar way).
Then by the property of $D^k_A$, $p$ has a child with label $[c+1, b_p]$ which still contains $[a,b]$, 
contradicting the fact that $p$ does not have such a child.
\qed
\end{proof}

By  Lemma~\ref{lem:dag}(iii), if the DAG $D^k_A$ and the answers for each sorted 2-sided \topk{} queries corresponding to all the nodes in $D^k_A$ are given, then we can answer any sorted $\topk{}(1,2, a, b, A)$ query by finding the corresponding node in $p$ in $D^k_A$ which satisfies $\topk{}(1,2, a, b, A) = \topk{}(p)$.

We now describe how to encode $D^k_A$ using at most $4n$ bits. 
The main idea of our encoding is as follows. For each node in $D^k_A$, we assign at most 2 bits (except the root node, which is assigned $k$ bits) while traversing all the nodes in the level order. These bits enable us to answer \topk{} queries on the range corresponding to each node in $D^k_A$, using the $\topk{}$ answers of the nodes in the previous level (more specifically, one of the parent nodes), and the $\topk{}$ encodings of the individual rows. By Lemma~\ref{lem:dag}(iii), this encoding is enough to answer all possible $\topk{}(1, 2, a, b, A)$ queries for any $1 \le a, b \le n$. However, since there exists at most $O(kn)$ nodes in $D^k_A$ (see Lemma~\ref{lem:numdag}), this encoding takes $O(kn)$ bits. To make the space independent to $k$, we skip some {\em redundant nodes} in $D^k_A$ (i.e., nodes for which the answers of $\topk{}$ on that nodes can be answered using the information obtained by some of the already traversed nodes, without any extra information). 
We modify the original level order to \textit{modified level order}, which will be describe later, and show that if we encode $D^k_A$ according to the modified level order, we can encode $D^k_A$ at most $4n$ bits, by skipping the redundant nodes during the traversal. 
\\\\
\noindent\textbf{Modified level-order.} For two nodes $p_i$ with label $I_{p_i} = [a_i, b_i]$ and $p_j$ with label  $I_{p_j} = [a_j, b_j]$ which satisfy $I_{p_i} \not\subset I_{p_j}$ and $I_{p_j} \not\subset I_{p_i}$, we say the node $p_i$ \textit{precedes} the node $p_j$ if $a_i< a_j$. 
Now, let $q$ be one of the parents of the node $p$ with label $I_p = [a,b]$ (note that a node can have multiple parents in a DAG). 
Note that $1 \le |\topk{}(q)-\topk{}(p)| \le 2$, since $I_p$ contains all the answers of $\topk{}(q)$ except one or both positions from the column $a-1$ or from the column $b+1$. Also let $f_p$ and $s_p$ be the number of positions in $\topk{}(q) \cap \topk{}(p)$ on the first and the second row respectively. Now we consider the following two cases: 

\begin{itemize}
    \item Case 1 ($|\topk{}(q)-\topk{}(p)|$ = 1) : In this case, the positions of $\topm{}(p)$ are already contained in the answers of $\topk{}(q)$, and the $k$-th largest element in $A[1,2][a, \dots, b]$ is either the $(f_{p}+1)$-th largest element in $A_1[a, \dots, b]$ or the $(s_{p}+1)$-th largest element in $A_2[a, \dots, b]$ (we call them as \textit{first-candidates} at node $p$).
    \item Case 2 ($|\topk{}(q)-\topk{}(p)|$ = 2) : In this case, the positions of $\tops{}(p)$ are already contained in the answers of $\topk{}(q)$, and the $(k-1)$-th largest element in $A[1,2][a, \dots, b]$ is the one of the the first-candidates at node $p$. Now suppose $(k-1)$-th largest element in $A[1,2][a, \dots, b]$ is on the first row (the other case is analogous). Then again, the $k$-th largest element in $A[1,2][a, \dots, b]$ is on the one of the positions of $(f_{p}+2)$-th largest element in $A_1[a, \dots, b]$ and $(s_{p}+1)$-th largest element in $A_2[a, \dots, b]$ (we call them as \textit{second-candidates} at node $p$). 
\end{itemize}

Note that if $f_p$ and $s_p$ are given, the first and second-candidates at node $p$ can be found using the $\topk{}$ encodings of $A_1$ and $A_2$. Figure~\ref{fig:modified} shows the overall procedure of modified level-order. While traversing the nodes of $D^k_A$ in the modified level order, we classify the nodes as \textit{visited}, \textit{half-visited}, or \textit{unvisited}. All the nodes are initially unvisited, and the traversal continues until all the nodes in $D^k_A$ are visited. For example, we traverse the nodes of $D^3_A$ in Figure~\ref{fig:dag} as:
$[1,9]\rightarrow[1,4]\rightarrow[1,4]\rightarrow[3,9]\rightarrow[1,2]\rightarrow[1,2]\rightarrow[3,6]\rightarrow[6,9]\rightarrow[6,9]
\rightarrow[1,1]\rightarrow[2,2]\rightarrow[3,4]\rightarrow[4,6]\rightarrow[6,7]
\rightarrow[8,9]\rightarrow[8,9]\rightarrow[3,3]\rightarrow[4,4]\rightarrow[5,6]\rightarrow[8,8]\rightarrow[9,9]\rightarrow[5,5]\rightarrow[6,6]$ (here each node is denoted as its label).

\begin{figure}[htbp]
\begin{framed}
 \begin{enumerate}
	\item{} Mark the root of $D^k_A$ as visited, and add its children into \textit{visit-list}, which is an ordered list such that for two nodes $p$ and $q$ in visit-list, $p$ comes before $q$ in visit-list if and only if $l(p) < l(q)$ or $l(p) = l(q)$ and $p$ precedes $q$ in the DAG ($l(p)$ denotes the level of the node $p$ which is defined as the number of edges in the longest path from root to $p$ in $D^k_A$).
	\item{} Find the leftmost unvisited or half-visited node $p$ from visit-list which satisfies one of the following conditions (without loss of generality, assume that $x \le y$).
	\begin{enumerate}
		\item{} Number of first or second candidates of $p$ is less than 2.
		\item{} First or second candidates of $p$ are $(1,x)$ and $(2, y)$, 
		and there exists no node $p'$ in visit-list such that (a) $I_p \subset I_{p'}$, or (b) $p'$ precedes $p$ and $x \in I_{p'}$, or (c) $I_p$ precedes $I_{p'}$ and $y \in I_{p'}$.
	\end{enumerate}

	Then we continue the traversal from $p$.
	\item{} Let $q$ be a parent of $p$. If (i) $|\topk{}(q)-\topk{}(p)| = 1$, or (ii) $|\topk{}(q)-\topk{}(p)| = 2$ and $p$ is half-visited, or (iii) the number of first or second candidates of $p$ is less than 2, then mark $p$ as visited, delete $p$ from the visit-list, and insert $p$'s children (if any) to visit-list. If none of these three conditions hold, then mark $p$ as half-visited.
	\item{} Repeat Steps 2 and 3 until all the nodes in $D^k_A$ are marked as visited.
\end{enumerate}
\end{framed}
\caption{Modified level-order traversal of $D^k_A$}
\label{fig:modified}
\end{figure}

\noindent\textbf{Picking the positions.} 
For a node $p = [a, b] \in D^k_A$, $p$ \textit{picks} the position $(x, y)$ if (i) $(x,y)$ is among the $\topk{}$ positions of $p$, and (ii) this information (that $(x,y)$ is among the answers to $\topk(p)$ query) does not follow from the 
$\topk{}$ positions of any of the visited or half-visited before $p$ in the modified level-order. By storing the information of all picked positions at node $p$, we can answer $\topk{}(p)$ by combining the answers of $\topk{}$  positions of some visited or half-visited nodes before $p$.

When the traversal starts at the root node of $D^k_A$, 
the root node picks the positions of $k$-largest values among the answers of $\topk{}(1, n, A_1)$ and $\topk{}(1, n, A_2)$ queries, and theses positions can be indicated using $k$ bits, since we assume that $\topk{}$ encodings of $A_1$ and $A_2$ are given.

Next, suppose we visit an unvisited non-root node $p$ where $|\topk{}(q)-\topk{}(p)| = 1$, where $q$ is a parent of $p$ (note that $q$ is always visited before $p$ in modified level order). In this case, since we can answer the positions of $(k-1)$ largest elements in $A$ using $\topk{}(q)$, $p$ picks at most one position, which is among the first-candidates at $p$. 
Thus, we can store the picked position at node $p$ using one extra bit. The case when $|\topk{}(q)-\topk{}(p)| = 2$ can be handled similarly, other than $p$ picks at most two positions (one from the first-candidates and another from the second-candidates), and this information can be stored using at most two extra bits. The following lemma shows that the size of $D^k_A$ is $O(kn)$, which in turn gives a simple $O(kn)$-bit space bound by storing the information of all the picked positions at each node of $D^k_A$.

\begin{lemma}\label{lem:numdag}
	Given $2 \times n$ array $A$ and DAG $D^k_A$,
	there are at most $6kn$ nodes in $D^k_A$. 
\end{lemma}
\begin{proof} 
	It is enough to show that there are at most $2kn$ non-leaf nodes in $D^k_A$. 
	Let $P_{[\alpha,\beta]} = \{I_{p_1} = [a_1, b_1], I_{p_2}=[a_2, b_2] \dots I_{p_t}=[a_t, b_t]\}$ be a set of labels of $t$ non-leaf nodes $p_1, p_2, \dots, p_t$ in $D^k_A$ where all the nodes $p_1, p_2, \dots, p_t$ picks $(\alpha, \beta)$. 
	Now we claim that $t$ is at most $k$.
	To prove a claim, suppose $t = k+1$.
	Then it is clear that for any $I_{p_i}, I_{p_j} \in P_{[\alpha,\beta]}$, $I_{p_i} \not\subset I_{p_j}$ and $I_{p_j} \not\subset I_{p_i}$ by the modified level-order traversal of $D^k_A$.
	Therefore without loss of generality, assume that $a_1 < a_2 \dots < a_{k+1}< \beta < b_{1} < b_{2} \dots < b_{k+1}$.
	Also we can easily show that for $1 \le i <k+1$, there is a position at $(b_i+1)$-th column whose corresponding value is larger than $A[p_{i_r}^k][p_{i_c}^k]$ by the construction algorithm of $D^k_A$.
	Therefore for $1 \le i \le  k$, the node $A[1,2][a_{k+1} \dots b_{k+1}]$ has $k$ positions at $(b_i+1)$-th columns which have larger values than both $A[\alpha][\beta]$, contradicts to the fact that $(\alpha,\beta) \in \topk{}(p_{k+1})$.
	\qed
\end{proof}

We now describe how to make the space usage of our encoding to be independent of $k$ - from $O(kn)$ to $O(n)$. Suppose there exists two non-root nodes $p$ and $q$ in $D^k_A$ where the first (or second) candidates of $p$ are contained in $q$, and the candidates of $p$ and $q$ are not distinct. In this case, the modified level order always visits $q$ prior to $p$, and gives a `chance' not to pick any position at node $p$, although $q$ is not an ancestor of $p$. Using this property, we now prove the following lemma, which bounds the size of our encoding by showing that if we store the all picked positions according to the modified level order, we can encode $D^k_A$ in space independent to $k$.

\begin{lemma}\label{lem:pick}
	Given $2 \times n$ array $A[1,2][1 \dots n]$ and DAG $D^k_A$, any position in $A$ is picked at most twice while we traverse all nodes in $D^k_A$ in the modified level order.
\end{lemma}
\begin{proof}
	Suppose that a position $(i, j)$ is among the answers of $\topk{}$ query on the $t$ distinct nodes $p_1, p_2, \dots, p_t$ where $I_{p_i}=[a_i, b_i]$ for $i \in \{1, 2, \dots, t\}$, but not among the answers of $\topk{}$ query on their parent nodes. 
	For $1 \le a, b \le t$ if $p_a$ is an ancestor of $p_b$, we don't pick $(i, j)$ at the node $p_b$ 
	by the modified level order traversal algorithm (Note that $p_a$ is traversed before $p_b$).
	Therefore without loss of generality, we assume that for all $1 \le a \neq b \le t$, $I_{p_a} \not\subset I_{p_b}$ and $I_{p_b} \not\subset I_{p_a}$. 
	Now we claim that for every position $(i, j)$ in $A$, at most two nodes from $p_1, \dots ,p_t$ pick $(i, j)$.
	To prove the claim, for $1 \le a < b < c \le t$, suppose there exists three nodes $p_a$, $p_b$, and $p_c$ where all of these three nodes picks $(i, j)$, and let $(i', j')$ be the another (first or second)-candidate of $p_b$. Then by the modified level-order traversal algorithm we do not pick $(i, j)$ at $p_b$ (note that $j' \in I_{p_a}$ or $j' \in I_{p_b}$), which contradicts the assumption.
\qed
\end{proof}

For example, during the traversal of $D^3_A$ in Figure~\ref{fig:dag} according to modified level order, the position(s) picked at each node are:
$\{(1,2),(1,5),(2,5)\}\rightarrow(1,3)\rightarrow(2,3)\rightarrow(2,7)
\rightarrow(1,1)\rightarrow(1,2)\rightarrow \epsilon \rightarrow(1,7)\rightarrow(1,8)
\rightarrow \epsilon \rightarrow \epsilon \rightarrow(2,4)\rightarrow \epsilon \rightarrow (1,6)
\rightarrow (1,9) \rightarrow (2,9) \rightarrow \epsilon\rightarrow \epsilon \rightarrow \epsilon \rightarrow \epsilon \rightarrow \epsilon
\rightarrow \epsilon \rightarrow \epsilon$, respectively ($\epsilon$ indicates that no position is picked).
Now we prove our main theorem. 

\begin{theorem}\label{thm:2n}	
	Given a $2 \times n$ array $A$, if there exists an $S(n, k)$-bit encoding to answer sorted 2-sided $\topk{}$ queries on a 1D array of size $n$ in $T(n,k)$ time and such encoding can be constructed in $C(n,k)$ time, then  we can construct an encoding of $A$ that uses $2S(n, k) + 4n$ bits which can be constructed in $O(C(n,k)+k^2n^2+knT(n,k))$ time, for answering $\topk{}$ queries on $A$.
\end{theorem}
\begin{proof}
	For $1 \le i \le j \le n$, we first use $2S(n, k)$ bits to support sorted 2-sided $\topk{}(1,1,i, j, A)$ and $\topk{}(2,2,i,j,A)$ queries. 
	To answer $\topk{}(1,2,i,j,A)$ queries, from Lemma~\ref{lem:dag}(iii), we note that it is enough to encode the answers to the sorted \topk{} queries corresponding to all the nodes in $D^k_A$.
	We encode these answers into a bit string $X$ while traversing the DAG $D^k_A$ as follows. 
	When the traversal begins at the root, $X$ is initialized to a $k$-bit string, which stores information for answering $\topk{}(1,2,1,n,A)$ query (namely, the $i$-th bit stores $0$ or $1$ depending on whether the $i$-th largest element in the range comes from the top or bottom row, respectively). 
	Now we traverse $D^k_A$ in the modified level order from the root node. Whenever we find a node $p$ in Step (2) of the traversal algorithm described above, and if we pick a position $(x,y)$ at node $p$, we append a single bit to $X$ to find the answer from $p$'s first (or second) candidate. 
	
	Note that we can find such $p$ in $O(kn)$ time by the Lemma~\ref{lem:numdag} and find the first (or second) candidates of node $p$ in $T(n,k)+O(1)$ time using the encoding of \topk{} queries on individual rows and pre-visited nodes other than root node, which takes $T(n,k)+O(k\lg{k})$ time. Finally we can check whether one of the position in the first (or second candidates) of $p$ is picked at node $p$ or not in $O(n)$ time. Therefore whenever we traverse node $p$, $O(kn+T(n,k))$ time is sufficient for encoding a bit in $X$ to find the answer of $\topk{}(p)$ query. 
	Since we traverse any node at most twice in the modified level order, and since $D^k_A$ has at most $6kn$ nodes by Lemma~\ref{lem:numdag}, we can construct the encoding in $O(C(n,k)+k^2n^2+knT(n,k))$ time in total.
	Also by Lemma~\ref{lem:pick}, $|X| \le 4n$ after we traverse all the nodes in $D^k_A$.
	
	To decode answers of $\topk{}$ queries corresponding to the nodes in $D^k_A$ from $X$, we first construct the root and its children from the first $k$ bits, and whenever we find a node $p$  with label $I_p = [a,b]$ in Step (2) of the traversal algorithm described above, we decode $(p^{k}_r, p^{k}_c)$ when $p$ is unvisited and $|\topk{}(q)-\topk{}(p)| = 1$, or $p$ is half-visited and $|\topk{}(q)-\topk{}(p)| = 2$. Also we decode $(p^{k-1}_r, p^{k-1}_c)$ when $p$ is unvisited and $|\topk{}(q)-\topk{}(p)| = 2$. 
	The positions in $\topk{}(p)$ with larger positions can be easily answered by the answer of $\topk{}$ on the former traversed nodes. 
	Now let $(1,x)$ and $(2,y)$ be the first or second candidates in such unvisited or half-visited node $p$, which can be found by $\topk{}(a,b,A_1)$ and $\topk{}(a,b,A_2)$.
	If one of the candidates is already picked before at some node $p'$ (without loss of generality, assume that $(1,x)$ is picked by $p'$) and $y \in p'$, we can know $A[1][x] > A[2][y]$ with no extra information. If there is no such node we read next 1 bit to decode. Since $X$ is encoded in the modified level order, one can easily show that such bit is encoded for pick $(i,j)$ at $p$. Therefore, we can encode to answer $\topk{}(1,2,i,j,A)$ queries at most $4n$ extra bits, if we can answer sorted 2-sided $\topk{}$ queries on $A_1$ and $A_2$.
	\qed
\end{proof}

For the special case when $k=2$, the following theorem shows that the space bound of the encoding of Theorem~\ref{thm:2n} can be improved.

\begin{theorem}\label{thm:top2_2n}	
	Given a $2 \times n$ array $A$, if there exists an $S(n)$-bit encoding to answer sorted 2-sided $\topt{}$ queries on a 1D array of size $n$ in $T(n)$ time and such encoding can be constructed in $C(n)$ time, then  we can encode $A$ in $2S(n) + 3n$ bits using $O(C(n)+n^2+nT(n))$ time, for answering $\topt{}$ queries on $A$.
\end{theorem}
\begin{proof}
It is enough to show that the bit string $X$, defined in the proof of Theorem~\ref{thm:2n}, has length at most $3n$ when it is constructed under $D^2_A$. 
We claim that after all nodes in $D^2_A$ are traversed in modified level-order, $i$-th column is picked (i.e., any position in the $i$-th column is picked) at most three times for all $1 \le i \le n$, which proves the theorem. 

To prove the claim, let $f(i)$ (respectively, $s(i)$) be the position of the larger (respectively, smaller) element between $(1, i)$ and $(2, i)$, and suppose $p$ with label  $I_p = [a,b]$ be the first node in the modified level order at which the position $s(i)$ is picked for the first time. 
Then by the traversing algorithm and definition of $D^2_A$, $f(i)$ is already picked before $s(i)$ is picked.
This implies that $\topt(p) = \{f(i), s(i)\}$ and the $i$-th column is not contained in all descendants of $p$ (note that the labels of $p$'s children are $[a, i-1]$ and $[i+1, b]$ by the definition of $D^2_A$).
Also, we claim that  $s(i)$ is not picked at any other node $p'$ with label $I_{p'}= [a', b']$  where $I_{p} \not\subset I_{p'}$ and $I_{p'} \not\subset I_{p}$. 
To prove this, suppose that $s(i)$ is picked at the node $p'$, and without loss of generality, $p'$ precedes the node $p$.    
Then by the definition of $D^2_A$, the element in $f(a-1)$ is larger than $f(i) \in \topt{}(p') \cap \topt{}(p)$ (note that the ancestor of $p$ picks $f(a-1)$, to have $p$ as descendant). This implies $s(i) \not\in \topt(p')$ and hence $s(i)$ cannot picked at the node $p'$.  
Thus, $s(i)$ is only picked once at the node $p$, and $f(i)$ can be picked at most twice by Lemma~\ref{lem:pick}, which implies any column is picked at most three times.
\qed
\end{proof}

From the encoding of Theorem~\ref{thm:2n}, the following theorem shows that we can obtain an encoding for answering sorted 4-sided $\topk{}$ queries on an $m \times n$ array by extending the encoding of a $2 \times n$ array.
\begin{theorem}\label{thm:mn}
	Given an $m \times n$ array $A$, if there exists an $S(n, k)$-bit encoding to answer sorted 2-sided $\topk{}$ queries on a 1D array of size $n$, then we can encode $A$ in $mS(n, k) + 2nm(m-1)$ bits, to support sorted 4-sided $\topk{}$ queries on $A$. 
\end{theorem}
\begin{proof}
	For $1 \le i \le j \le n$ and $1 \le a \le b \le m$, we first use $mS(n, k)$ bits to support sorted 2-sided $\topk{}(a,a,i, j, A)$ queries. Also we encode the answer $\topk(1, 2, i, j, A_{ab})$ queries on $2 \times n$ array $A_{ab}$, whose first and second row are $a$-th and $b$-th row in $A$ respectively. By Theorem~\ref{thm:2n}, we can encode such queries on all possible $A_{ab}$ arrays using $2nm(m-1)$ extra bits. 
	For $1 \le a \le m$ and $1 \le \ell \le k$, let $\alpha_{ij}(a,\ell)$ be the position of $\ell$-th largest element in $A[a][i \dots j]$ Note that we can find such $\alpha_{ij}(a,\ell)$ using $\topk{}(a,a,i,j, A)$ query. 
	
	Now we describe how to answer $\topk{}(a,b,i,j,A)$ query. We first define $(b-a+1)$ values $c_a \dots c_b$ and set $c_a = c_{a+1} \dots = c_b =1$. After that, we find a position of largest value in $A[a \dots b][i \dots j]$ by comparing $\alpha_{ij}(a,c_a), \alpha_{ij}(a+1,c_{a+1}) \dots \alpha_{ij}(b,c_{b})$ and find a position with largest element among them. 
	It is clear that for $a \le a' \le b' \le b$, we can compare values at the position $\alpha_{ij}(a',c_{a'})$ and $\alpha_{ij}(b',c_{b'})$ using $\topk{(1, 2, i, j, A_{a'b'})}$ query since at least one of the their corresponding positions in $A_{a'b'}$ is an answer of $\topk{(1, 2, i, j, A_{a'b'})}$ query. 
	Suppose for $a \le a' \le b$, $\alpha_{ij}(a',c_{a'})$ a position with the largest value in $A[a \dots b][i \dots j]$. Then we increase $c_{a'}$ by 1, and compare $\alpha_{ij}(a',c_{a'})$ and $\alpha_{ij}(b',c_{b'})$ again to find a position of the second largest value in $A[a \dots b][i \dots j]$. 
	We do this procedure iteratively until we find a position of $k$-th largest value in $A[a \dots b][i \dots j]$.
	\qed
\end{proof}

\begin{corollary}\label{col:top22n}
Given an $m \times n$ array $A$, if there exists an $S(n)$-bit encoding to answer sorted 2-sided $\topt{}$ queries on a 1D array of size $n$, then we can encode $A$ in $mS(n) + 1.5nm(m-1)$ bits, to support sorted 4-sided $\topt{}$ queries on $A$. 
\end{corollary}

Finally, if we combine the encoding of Theorem~\ref{thm:mn} and Gawrychowski and Nicholson's $(\lg{{(k+1)n \choose n}}+o(n))$-bit optimal encoding for sorted 2-sided  \topk{} queries on a 1D array~\cite{DBLP:conf/icalp/GawrychowskiN15}, we obtain an encoding as follows.

\begin{corollary}\label{col:mn}
	Given an $m \times n$ array $A$, there exists an $(m\lg{{(k+1)n \choose n}}+2nm(m-1)+o(n))$-bit encoding, to support sorted 4-sided $\topk{}$ queries on $A$.  Also when $k=2$,  there exists an $(m\lg{{3n \choose n}}+1.5nm(m-1)+o(n))$-bit encoding, to support sorted 4-sided $\topt{}$ queries on $A$
\end{corollary}
\section{Data structure for 4-sided $\topk{}$ queries on $2 \times n$ array. }\label{sec:2nds}
The encoding of Theorem~\ref{thm:2n} shows that $4n$ bits are sufficient for answering $\topk{}$ queries whose range spans both rows, when encodings for answering sorted 2-sided $\topk{}$ queries for each row are given. 
However, this encoding does not support queries efficiently (takes $O(k^2n^2+knT(n,k))$ time) since we need to reconstruct all the nodes in $D_A$ to answer a query (in the worst case). 
We now show that the query time can be improved to $O(k^2+kT(n,k))$ time 
if we use $(4k+7)n+ko(n)$ additional bits.
Note that if we simply use the data structure of Grossi et al.~\cite{DBLP:journals/talg/GrossiINRR17} 
(which takes $44 n\lg{k}+O(n\lg{\lg{k}})$ bits to encode a 1D array of length $n$ to support $\topk{}$ queries in $O(k)$ time)
on the 1D array of size $2n$ obtained by writing the values of $A$ in column-major order, 
we can answer $\topk{}$ queries on $A$ in $O(k)$ time using $88 n\lg{k}+O(n\lg{\lg{k}})$ additional bits. 
Although our data structure takes more query time 
and takes asymptotically more space, 
it uses less space for small values of $k$ (note that $4k+7 < 88\lg{k}$ for all integers $2 \le k < 160$) when $n$ is sufficiently large.
We now describe our data structure.

\begin{figure}[t]
\begin{center}
\includegraphics[clip, width=12cm]{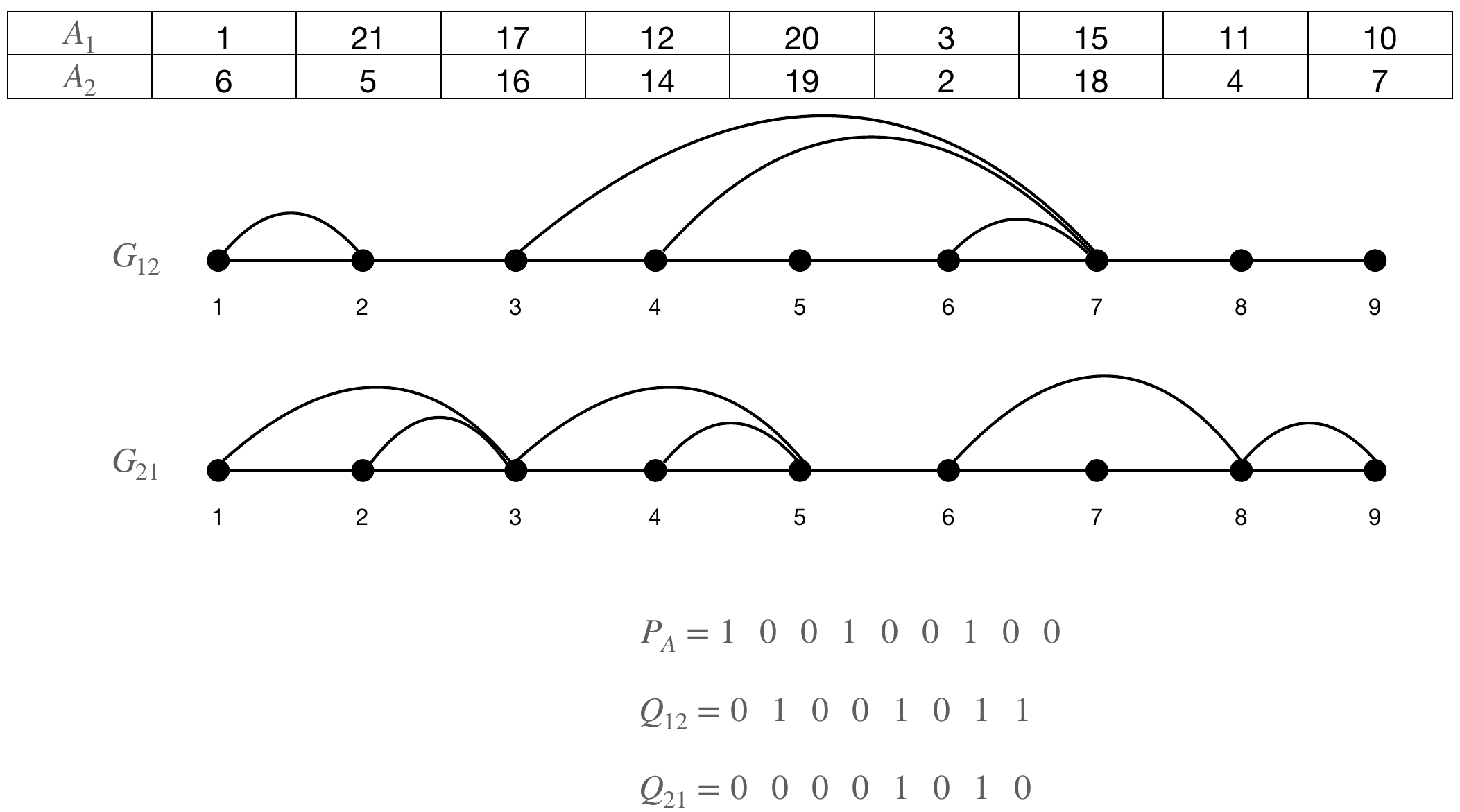}
\caption{$G_{12}$, $G_{21}$, $P_A$, $Q_{12}$, and $Q_{21}$ to support $\textsf{Top-}3$ queries on $2 \times n$ array $A$
}
\label{fig:kgraph}
\end{center}
\end{figure}

We first define a graph $G_{12} = (V(G_{12}), E(G_{12}))$ on $A$ as follows.
The set of vertices $V(G_{12})=\{1, 2, \dots n\}$, and there exists an edge $(i, j) \in E(G_{12})$
if and only if (i) $i < j$ and $A[1][i] < A[2][j]$, 
(ii) there are at most $k-1$ positions in $A[1,2][i \dots j]$ whose corresponding values are larger than both $A[1][i]$ and $A[2][j]$, 
and (iii) there is no vertex $j' \in \{i+1, i+2, \dots, n\}$ where $A[1][i] < A[2][j'] < A[2][j]$ and satisfies the condition (ii).
We also define a graph $G_{21}$ on $A$ which is analogous to $G_{12}$, by replacing $A[1][i]$, $A[2][j]$ and $A[2][j']$ with $A[2][i]$ $A[1][j]$, and $A[1][j']$, respectively in all three conditions. 
Each of the graphs $G_{12}$ and $G_{21}$ have $n$ vertices and at most $n$ edges. Also for any vertex $v \in V(G_{12})$ (resp., $V(G_{21})$), there exists at most one vertex $v'$ in $G_{12}$ (resp., $G_{21}$) such that $v$ is incident to $v'$ and $v < v'$. See Figure~\ref{fig:kgraph} for an example.
We now show that $G_{12}$ (thus, also $G_{21}$) 
is a \textit{$k$-page graph}, i.e. there exist no $k+1$ edges 
$ (i_1, j_1) \dots (i_{k+1}, j_{k+1}) \in E(G_{12})$ such that 
$i_1 < i_2  \dots <i_{k+1} < j_1 < j_2 \dots <j_{k+1}$.

\begin{lemma}\label{lem:kpage}
	Given $2 \times n$ array $A$, a graph $G_{12}$ on $A$ is $k$-page graph.
\end{lemma}
\begin{proof}
	Suppose that there are $k+1$ edges $ (i_1, j_1) \dots (i_{k+1}, j_{k+1}) \in E(G_{12})$ such that 
	$i_1 < i_2  \dots <i_{k+1} < j_1 < j_2 \dots <j_{k+1}$, and 
	for $1 \le t \le k+1$, let $i_t$ be a position of the minimum element in $A_1[i_1 \dots i_{k+1}]$.
	Then by the definition fo $G_{12}$, there are at least $k$ positions 
	$(1,i_{t+1}), \dots, (1, i_{k+1}), (2,j_1), \dots, (2,j_{t-1}) $ in $A[1,2][i_t \dots j_t]$ 
	whose corresponding values in $A$ are larger than both 
	$A[1][i_t]$ and $A[2][j_t]$, which contradicts the definition of $G_{12}$.
	\qed
\end{proof}
From the above lemma and the succinct representation of $k$-page graphs of 
Munro and Raman~\cite{mr-sjc01} (with minor modification as described in~\cite{DBLP:conf/cpm/GawrychowskiN15}), 
we can encode $G_{12}$ and $G_{21}$ using $(4k+4)n+k \cdot o(n)$ bits in total, and
for any vertex $v$ in $V(G_{12}) \cup V(G_{21})$, we can find a vertex with the largest index
which incident to $v$ in $O(k)$ time.
Also to compare the elements in the same column, we maintain a bit string $P_A[1 \dots n]$ of size $n$ such that for $1 \le i \le n$, $P_A[i] = 0$ if and only if $A[1][i] > A[2][i]$.
Finally, for $G_{12}$ (resp., $G_{21}$), we maintain 
another bit string $Q_{12}[1 \dots n-1]$ (resp., $Q_{21}[1 \dots n-1]$) 
such that for $1 \le i \le n-1$, $Q_{12}[i] = 1$ (resp., $Q_{21}[i] = 1$) if and only if all elements in $A_2[i+1 \dots n]$ (resp., $A_1[i+1 \dots n]$) are smaller than
$A[1][i]$ (resp., $A[2][i]$) (see Figure~\ref{fig:kgraph} for an example).
We now show that if there is an encoding which can answer the sorted $\topk{}$ queries on each row, then the encoding of $G_{12}$, $G_{21}$, and 
the additional arrays defined above are enough to answer 4-sided $\topk{}$ queries on $A$.

\begin{theorem}\label{thm:ds2n}
	Given a $2 \times n$ array $A$, if there exists an $S(n, k)$-bit encoding to answer sorted 2-sided $\topk{}$ queries on a 1D array of size $n$ in $T(n,k)$ time, then there is a $2S(n, k) + (4k+7)n+k \cdot o(n)$-bit data structure which can answer $\topk{}$ queries on $A$ in $O(k^2 + kT(n,k))$ time.
\end{theorem}
\begin{proof}
	For $1 \le i \le j \le n$, we first use $2S(n, k)$ bits to support sorted 2-sided $\topk{}(1,1,i, j, A)$ and $\topk{}(2,2,i,j,A)$ queries in $T(n,k)$ time. 
	To answer $\topk{}(1,2,i,j,A)$ query,
	we maintain succinct representations of $G_{12}$ and $G_{21}$~\cite{mr-sjc01,DBLP:conf/cpm/GawrychowskiN15}
	using $(4k+4)n+ko(n)$ bits, and $P_A$, $Q_{12}$, and $Q_{21}$ using $3n$ bits. 
	Now for $1 \le p \le k$, let $a_p$ (resp., $b_p$) be the position of the $p$-th largest value in $A_1[i \dots j]$ 
	(resp., $A_2[i \dots j]$), which can be answered in $O(T(n,k))$ time using the encoding of $\topk{}$ queries on each row. 
	We first find the position of the largest value in $A[1,2][i \dots j]$ 
	by comparing $A[1][a_1] $ and $A[2][b_1]$.
	If $A[1][a_1] < A[2][b_1]$(resp.,$A[1][a_1] > A[2][b_1]$), we compare $A[1][a_1]$ with $A[2][b_2]$ (resp., $A[1][a_2]$ with $A[2][b_1]$) 
	to find the position of the second-largest value in $A[1,2][i \dots j]$. 
	By repeating this procedure iteratively $k$ times, we can answer the $\topk{}(1,2,i,j,A)$ query.
	
	Now we describe how to compare $A[1][a_p]$ with $A[2][b_q]$, for all $p+q \le k+1$ 
	(in the above procedure, we do not need to compare $A[1][a_p]$ with $A[2][b_q]$ if $p+q > k+1$). 
	If $a_p = b_q$, the result of the comparison is already stored in the bit $P_A[a_p]$. 
	Now suppose that $a_p < b_q$ (if $a_p > b_q$, we use $G_{21}$ and $Q_{21}$ instead of $G_{12}$ and $Q_{12}$ respectively, in the following procedure), 
	and let $a'_p$ be a vertex with the largest index in $G_{12}$ which is incident to $a_p$, if it exists.
	Note that we can find such $a'_p$  in $O(k)$ time~\cite{mr-sjc01,DBLP:conf/cpm/GawrychowskiN15}.
	If there is no vertex incident to $a_p$ or $a'_p < a_p$, we show that $A[1][a_p] > A[2][b_q]$ by considering the following two cases.
	\begin{itemize}
		\item {i) $Q_{12}[a_p] =  1$:}
		From the definition of $Q_{12}$, it follows that $A[1][a_p] > A[2][b_q]$.
		\item {ii) $Q_{12}[a_p] =  0$:}
		In this case, (a) $A[1][a_p] < A[2][b_q]$, but 
		there are at least $k$ positions in $A[1,2][a_p \dots b_q]$ whose corresponding values are larger than both $A[1][a_p]$ and $A[2][b_q]$
		or (b) $A[1][a_p] > A[2][b_q]$. 
		However (a) cannot hold 
		since there are at most $(p-1)+(q-1) \le k-1$ positions in $A[1,2][a_p \dots b_q]$ whose corresponding values are larger than both $A[1][a_p]$ and $A[2][b_q]$. 
		Therefore $A[1][a_p] > A[2][b_q]$.
	\end{itemize}
	Now consider the case $a_p < a'_p$.
	If $a'_p \le b_q$, then $A[1][a_p] < A[2][b_q]$ if and only if $A[2][a'_p] < A[2][b_q]$ 
	by the definition of $G_{12}$.
	If $a'_p > b_q$, we first compare $A[2][a'_p] $ with $A[2][b_q]$. 
	If $A[2][a'_p] < A[2][b_q]$, then $A[1][a_p] < A[2][b_q]$ by the definition of $G_{12}$. 
	If not, (a) $A[1][a_p] < A[2][b_q]$, but there are at least $k$ positions in $A[1,2][a_p \dots b_q]$ whose corresponding value is larger than both $A[1][a_p]$ and $A[2][b_q]$, or (b) $A[1][a_p] > A[2][b_q]$. However, (a) cannot hold by the same reason as the case 
	when there is no vertex incident to $a_p$ or $a'_p < a_p$, and $Q_{12}[a_p] =  0$. 
	Therefore $A[1][a_p] > A[2][b_q]$ if $A[2][a'_p] > A[2][b_q]$. 
	Also since $(2, b_q)$ is one of the answers of $\topk{}(i,j, A_2)$ query, 
	we can compare $A[2][a'_p]$ with $A[2][b_q]$ in $T(n,k)$ time using the $\topk{}$ encoding on 
	the second row.
	By the procedure describe above, 
	each iteration step takes at most $O(k+T(n,k))$ time, thus we can answer
	$\topk{}(1,2,i,j.A)$ query in $O(k^2+kT(n,k))$ time. 
	\qed
\end{proof}

\section{Lower bounds for encoding range \topk{} queries on $2 \times n$ array}
\label{sec:lower}
In this section, we consider the lower bound on space for encoding a $2 \times n$ array $A$ to support unsorted 1-sided and sorted 4-sided $\topk{}$ queries, when $k > 1$. 
Specifically for $1 \le i \le j \le n$, we consider to lower bound on extra space for answering 
i) unsorted $\topk{}(1, 2, 1, i)$ queries, assuming that we have access to the encodings of the individual rows of $A$ that can answer 
unsorted 1-sided (or 2-sided) $\topk{}$ queries, and ii) sorted $\topk{}(1, 2, i, j)$ queries, assuming that we have access to the encodings of the individual rows of $A$ that can answer sorted 2-sided $\topk{}$ queries.
We show that for answering unsorted 1-sided queries on $A$, at least $1.27n-o(n)$ (or $2n-O(\lg{n})$) extra bits are necessary, and for answering unsorted or sorted 4-sided queries on $A$, at least $2n-O(\lg{n})$ extra bits are necessary.

For simplicity (to avoid writing floors and ceilings, and to avoid considering some boundary cases), we assume that $k$ is even. 
(Also, if $k$ is odd we can consider the lower bound on extra space for answering 4-sided $\topk{}$ queries as the lower bound of extra space for answering 4-sided $\topm{}$ queries -- it is clear that former one requires more space.) 
For both unsorted and sorted query cases, we assume that all elements in $A$ are distinct, and come from the set $\{1, 2, \dots 2n\}$; and also that each row in $A$ is sorted in the ascending order. 
Finally, for $1 \le \ell \le 2n$, we define the mapping $A^{-1}(\ell) = (i, j)$ if and only if $A[i][j] = \ell$.\\

\noindent
{\bf Unsorted 1-sided \topk{} query. }
The following theorem gives the lower bound for answering unsorted 1-sided $\topk{}$ queries on $2 \times n$ array $A$, when the encodings for answering unsorted 1-sided (or 2-sided) \topk{} queries on both rows are already given. Note that this lower bound also gives the lower bound for answering 
unsorted 4-sided $\topk{}$ queries on $2 \times n$ array under the same condition.

\begin{theorem}\label{lower:unsorted}
	Given a $2 \times n$ array $A$ and encodings for answering unsorted 1-sided (or 2-sided) \topk{} queries on both rows in $A$, at least $\ceil{(n-k/2)\lg{(1+\sqrt{2})}}-o(n) = 1.27(n-k/2)-o(n)$ additional bits are necessary for answering unsorted 1-sided \topk{} queries on $A$. 
\end{theorem}
\begin{proof}
If $n \le k/2$ we do not need any extra space since all positions are answers of unsorted $\topk{}(1, 2, 1, i, A)$ queries for $i \le n$. Now suppose that $n > k/2$.
In this case, let $U_i$ be a set of all possible arrays of size $2 \times n$ which satisfies the following properties:

\begin{itemize}
    \item For any $B \in U_i$, all of $\{1, 2 \dots 2i\}$ are in $B[1,2][1 \dots i]$ and each row in $B$ is sorted in the ascending order (thus, all the arrays in $U_i$ have same encodings for answering unsorted 1 and 2-sided $\topk{}$ queries on their individual rows), and
    \item for any two distinct arrays $B , C \in U_i$, there exists $1 \le j \le i$ such that $\{B^{-1}(2j-1),B^{-1}(2j)\} \neq \{C^{-1}(2j-1), C^{-1}(2j)\}$. 
\end{itemize}

By the definition of $U_i$, for any $1 \le i \le n-k/2$ and two distinct arrays $B , C \in U_i$, there exists a position $1 \le j \le i$ where $B$ and $C$ have distinct answers of unsorted $\topk{}(1, 2, 1, k/2+j)$ queries, which implies $\log |U_n|$ gives the lower bound of additional space for answering the 1-sided $\topk{}$ queries on $2 \times n$ array.
We compute the size of $U_i$ as follows.
$|U_1| = 1$ since there exists only one case as $\{B^{-1}(1), B^{-1}(2)\} = \{(1,1), (2,1)\}$. 
For $i=2$, we can consider three cases as $(1,2,3,4)$, $(1,3,2,4)$, or $(1,4,2,3)$ if we write a elements of $B[1,2][1,2]$ for $B \in U_2$ in row-major order (note that each row is sorted in ascending order).

Next, consider the case when $2 < i \le n-k/2$. In this case for any $B \in U_{i-1}$, $\{B^{-1}(2i-3), B^{-1}(2i-2)\} \subset \{\{(1, i), (2, i)\}, \{(1, i-1), (1, i)\}, \{(2, i-1), (2, i)\}\} $. 
To construct arrays in $U_i$, we construct a set $E_B \subset U_i$ from the array in $B \in U_{i-1}$ such that for any $B_1 , B_2 \in E_B$ and $1 \le j < i+k/2$, $\topk{}(1,2,1,j,B_1) = \topk{}(1,2,1,j,B_2) = \topk{}(1,2,1,j,B)$ 
and $\topk{}(1,2,1,i+k/2,B_1) \neq \topk{}(1,2,1,i+k/2,B_2)$. It is clear that $|U_i| = \sum_{B \in U_{i-1}} |E_B|$. 
Now we consider two cases as follows.  

\begin{itemize}
	\item{}\textbf{Case 1. $B^{-1}(2i-3)$ and $B^{-1}(2i-2)$ are in different rows:}
	In this case, for any $C\in E_B$ the position of the first and the second largest value in $C[1,2][1 \dots 2(i-1)]$
	are $(1, i-1)$ and $(2, i-1)$ respectively and for any $B_1, B_2 \in E_B$, $\{B^{-1}_1(2i-1), B^{-1}_2(2i)\}$ are distinct.
	In this case, $|E_B| = 3$ since only $\{(1, i), (2, i)\}$, $\{(1, i-1), (1, i)\}$, or $\{(2, i-1), (2, i)\}$ can be $\{B^{-1}(2i-1), B^{-1}(2i)\}$  for any array $C \in E_B$ respectively, which satisfy the above condition and maintaining both rows in $C$ as sorted in ascending order.
	Furthermore, since both $B^{-1}(2i-3)$ and $B^{-1}(2i-2)$ are in $(i-1)$-th column, the number of $B \in U_{i-1}$ in this case is $|U_{i-2}|$. 
	
	\item{}\textbf{Case 2. $B^{-1}(2i-3)$ and $B^{-1}(2i-2)$ are in the same row:}
	Without loss of generality, assume that both $B^{-1}(2i-3)$ and $B^{-1}(2i-2)$ are in the first row. 
	Then for any $C\in E_B$ the position of the first and the second largest value in  $C[1,2][1 \dots 2(i-1)]$
	are $(1, i-2)$ and $(1, i-1)$ respectively, and for any $B_1, B_2 \in E_B$, $\{B^{-1}_1(2i-1), B^{-1}_2(2i)\}$ are distinct. 
	Therefore $|E_B| = 2$ in this case since only $\{(1, i), (2, i)\}$ or $\{(1, i-1), (1, i)\}$ can be $\{C(2i-1), C(2i)\}$ for any array $C \in E_B$ respectively 
	which satisfy the above condition and maintaining both rows in $C$ as sorted in ascending order. Also since all the $B$ is in the Case 1 or 2, the number of $B \in U_{i-1}$ in this case is $|U_{i-1} - U_{i-2}|$.
\end{itemize} 

By the statement described above, we obtain a recursive relation $|U_i| = 3|U_{i-2}| + 2(|U_{i-1}| - |U_{i-2}|)$. 
By solving the recursive relation using characteristic equation, we obtain $|U_i| \le (1+\sqrt{2})^i-o(n)$, which proves the theorem. 
\qed
\end{proof}

\noindent
{\bf Sorted and 4-sided \topk{} query. }
In this case we divide a $2 \times n$ array $A$ into $2n/k$ blocks $A_1 \dots A_{2n/k}$ of size $2 \times k/2$ as
for $1 \le \ell \le k/2$, $A_{\ell}[i][j] = A[i][2(\ell-1)+j]$ and all values of $A_{\ell}$ are in $\{k(\ell-1)+1 \dots k\ell\}$. Then for any $2 \times n$ array $A$ and $A'$, sorted $\topk{}(1,2,k(i-1)/2+1, ki/2, A) \neq \topk{}(1,2,k(i-1)/2+1, ki/2, A')$, and $\topk{}(1,2,1, ki/2, A) \neq \topk{}(1,2,1, ki/2, A')$, if there exists a position $1 \le i \le 2n/k$ where $A_i \neq A'_i$. 
Let $S_{i}$ be the set of all possible arrays of size $2 \times i$ such that for any $B \in S_{i}$, all values of $B$ are in $\{1, 2i\}$ and both rows of $B$ are sorted in ascending order, which implies all the arrays in $S_i$ have same encodings for answering unsorted and sorted 2-sided $\topk{}$ queries on their individual rows.
Since the size of $S_{i}$ is same as \textit{central binomial number}, 
${2i} \choose {i}$, which is well-known as at least $4^i/\sqrt{4i}$~\cite{binomial}. 
Therefore, at least $\ceil{2n\lg|S_{k/2}|/k} \ge 2n -O(\lg{n})$ additional bits are necessary for answering sorted $\topk{}$ queries that span both the rows, when encodings for answering sorted (or unsorted) on both rows are given.

\begin{theorem}\label{lower:sorted}
	Given a $2 \times n$ array $A$, at least $2n-O(\lg{n})$ additional bits are necessary for answering sorted 4-sided \topk{} queries on $A$ if encodings for answering sorted (or unsorted) 2-sided \topk{} queries on both rows in $A$ are given.
\end{theorem}

\section{Conclusions and open problems}
In this paper, we proposed encodings for answering $\topk{}$ queries on 2D arrays. 
For $2 \times n$ arrays, we proposed upper and lower bounds on space for answering sorted and unsorted 4-sided  $\topk{}$ queries. 
Finally, we obtained an $(m\lg{{(k+1)n \choose n}}+2nm(m-1)+o(n))$-bit encoding for answering 4-sided sorted $\topk{}$
queries on $m \times n$ arrays. 
We end with the following open problems: 
\begin{itemize}
\item{} Can we support 4-sided sorted $\topk{}$ queries with efficient query time on $m \times n$ arrays
using less than $O(nm\lg{n})$ bits when $m = o(\lg{n})$?
\item{} Can we obtain an improved lower or upper bound for answering 4-sided sorted $\topk{}$ queries on $2 \times n$ arrays?
\end{itemize}

\bibliography{ref}

\end{document}